\documentclass[a4paper,llncsarea,envcountsame,envcountsech,runningheads]{class/llncs}

\usepackage[USenglish]{babel}
\usepackage{xargs}
\usepackage{algorithm}
\usepackage{amsmath}
\usepackage{color}
\usepackage{xcolor}

\newcommand{\y}{\vphantom{$c'$}}

\usepackage{tikz}
\usetikzlibrary{arrows.meta,decorations,decorations.pathreplacing,calc,bending,positioning, 
	chains,patterns,scopes}

\tikzstyle{box}=[draw, minimum width=2em, minimum height=1.6em, outer 
sep=0pt, inner sep=0pt ,on chain]
\tikzstyle{dottedline}=[ultra thick, loosely dotted,shorten >=1mm, shorten <=1mm]

\usetikzlibrary{chains, fit, positioning}

\newcommand{\field}[4][white]{%
	\node [box,fill=#1] (#2) {#3};
	\node [above=0 of #2] (#2i) {#4};%
}

\newcommand{\fieldt}[3]{%
	\node [box] (#1) {#2};
	\node [above=0 of #1] (#1i) {#3};%
}

\def\nameFields#1#2#3{%
	\draw[decorate,decoration={brace, amplitude=5pt, raise=2pt, mirror}]
	(#1.south west) + (0,0) to node[black,midway,below=6pt] {#3} (#2.south east) + (0, 0);%
}

\def\sep#1{%
	\draw(#1.east) -- + (0, 8mm);
	\draw(#1.east) -- + (0,-8mm);
}

\usepackage[inline]{enumitem}

\newboolean{isConferenceVersion}

\setboolean{isConferenceVersion}{false}

\long\def\full#1{\ifthenelse{\boolean{isConferenceVersion}}{}{{{#1}}}}
\long\def\conf#1{\ifthenelse{\boolean{isConferenceVersion}}{{{#1}}}{}}
	
\newcommand{\Nat}{I\!\!N}

\definecolor{khaki}{RGB}{255,255,204}
\definecolor{darkgreen}{rgb}{0.0,0.5,0.0}
\definecolor{darkblue}{rgb}{0.0,0.0,0.7}
\definecolor{lightblue}{RGB}{102,204,255}
\definecolor{purple}{RGB}{230,230,250}
\definecolor{lighterblue}{RGB}{102,204,255}
\definecolor{lightestblue}{RGB}{153,234,255}
\definecolor{lightred}{RGB}{255,71,71}
\definecolor{lighterred}{RGB}{255,144,144}
\definecolor{lightestred}{RGB}{255,204,204}
\definecolor{notgreenish}{rgb}{0.7,0.5,0.0}
\definecolor{greenish}{rgb}{0.2,0.5,0.2}
\definecolor{lightgreen}{RGB}{153,255,102}
\definecolor{red}{RGB}{255,152,152}

\title{Linear-Time In-Place DFS and BFS on~the~Word~RAM}

\author{Frank Kammer and Andrej Sajenko}
\institute{THM, University of Applied Sciences Mittelhessen, 
	Germany\
	\email{\{frank.kammer,andrej.sajenko\}@mni.thm.de}}

\begin{document}
\thispagestyle{plain}
\pagestyle{plain}

\maketitle{}%

\begin{abstract}%
 We present an in-place depth first search (DFS) 
 and an in-place breadth first search (BFS)
that runs on a word RAM in linear time such that, if the adjacency arrays of the input graph are given in a sorted order, the input is restored after running the algorithm.
 To obtain our results we use properties of the representation used to store the given graph and show 
 several linear-time in-place graph
 transformations from 
 one representation into another.
%
%
\end{abstract}
\section{Introduction}

\thispagestyle{plain}
Motivated by the rapid growth of the data sizes in nowadays applications, algorithms that are 
designed to efficiently utilize both time and space are becoming more and more important.
Another reason for the need of such algorithms is the limitation in the memory sizes of the tiniest 
devices.

To measure the total amount of memory that an algorithm requires we distinguish two types of 
memory. 
The memory that stores the input is called the \textit{input memory}.
The memory that an algorithm additionally occupies during the computation is called the 
\textit{working memory}.

Several models of computation have been considered for the case when writing in the input memory is 
restricted. 
In the \emph{multi-pass streaming} model~\cite{MunP80} the input is assumed to be held in
a read-only sequentially-accessible media, and the \full{main }optimization target is the number of passes an 
algorithm makes over the input.
In the {\em word RAM}~\cite{Fre87} the memory is partitioned into
randomly-accessible {\em words}, each of size $w$, the input is in the 
first $N\in \Nat$ words and reading/writing a word as well as the arithmetic operations 
(addition, subtraction, multiplication and bit-shift) take constant time 
if applied on inputs that fit into a word. As usual, we assume 
$w =
\Omega(\log N)$. 
In the \emph{read-only word RAM}~\cite{Fre87} the input memory is assumed to be read-only.
Another model allows data in the input memory to be permuted, but not destroyed \cite{BroIKMMT04}.
A variant of the latter model is called the \emph{restore model} \cite{ChaMR14}
where the input memory is allowed to be modified during the process of answering a
query, \mbox{but it has to be restored to its original state afterwards.} 

There are several algorithms for the read-only word RAM, e.g., for
sorting~\cite{Bea91,PagR98},  
geometric problems~\cite{AsaBBKMRS14,BarKLSS13},
or 
graph
algorithms~\conf{\cite{AsaIKKOOSTU14,ChaRS16,DatKM16,ElmHK15,HagK16}}\full{\cite{AsaIKKOOSTU14,ChaRS16,DatKM16,ElmHK15,HagK16,HagKL18,KamKL16}}.
Unfortunately, most of the algorithms on $n$-vertex graphs (including {\em depth first search} (DFS) and {\em breadth first search} (BFS)) have to use 
roughly $\Omega(n)$ bits
of working memory in the read-only RAM model since there is a lower bound for the {\em reachability 
problem}, i.e., the problem to find out if two given vertices of a given graph are in the
same connected component. 
The lower bound essentially says that we can solve reachability
in polynomial time only if we have roughly $\Theta(n)$ bits of working memory~\cite{EdmPA99}. 

Our focus is to find {\em space-efficient algorithms}, i.e., algorithms that 1.) run (almost) as fast 
as the best known 
algorithms for the problem without any space limitations and that 2.) use
space economically.
To bypass the lower bound 
we consider \textit{in-place} algorithms.
An in-place algorithm~\cite{ChanMR18} can use the input memory and the working memory for
writing, and 
the result of the algorithm may be written to the input or can be 
sent to an output stream. Moreover, the working memory size is restricted to $O(1)$ words. 
Sorting algorithms like heapsort and bubblesort are classic examples of in-place
algorithms.
 
Usually, one runs several computations on a given graph. To allow the input to be reused 
\full{for further 
computations, }we want to run our algorithms on the {\em weak restore word RAM}, i.e., 
given the input in a specific representation, as for example the sorted representation in the next section, it can be restored.

Graph algorithms usually do not specify the input format 
of a given graph since
linear time \full{and a linear number of words in the
working memory are}\conf{is}   
sufficient to convert between any two   
reasonable adjacency-list representations---e.g.,
reorder the adjacency arrays with
radix sort. However, since we focus on
linear-time in-place algorithms for DFS and BFS in the 
weak restore word RAM,
we have to be more specific about the
input format. 
Implementing 
an in-place algorithm on the weak restore word RAM model where the working memory
is limited and the input memory must be restored, a trick is to use
the redundancy in the input representation.
Thus, the size of the input representation is very crucial.
In the following, let $n$ and $m$ be the number of vertices and edges,
respectively, of the given graph.

We are not aware of a linear-time DFS or BFS that runs in-place or uses this model.
However, Chakra\-borty et al.~\cite{ChaMRS17} introduced another model
where the adjacency arrays of a graph can be only rotated, but a restoration is not required. 
In their model, they recently showed that one can run an in-place DFS and a BFS in $O(n^3 \log n)$ 
time on an arbitrary graph.
The space required to represent the graph is not mentioned explicitly, 
but based on their description 
they require at least $(n + 2m+\min\{n,m/w\})$ words for undirected graphs
since each undirected 
edge 
is stored at both endpoints and since an adjacency array is used for each
vertex where the size of the array must be known.
Moreover, their representation for directed graphs uses
at least $(2n + 2m + 2\min\{n,m/w\})$ words since adjacency arrays for 
in- and 
out-edges 
are stored for each vertex.


We use the weak restore word RAM to show linear-time, in-place algorithms for both DFS and BFS that 
runs on a 
graph with a representation consisting of only $(n+m+2)$ words
on directed graphs and $(n+2m+2)$ words on undirected graphs (each
undirected edge occurs at both endpoints).
To operate efficiently on that compact representation and to have also some
kind of redundancy, we assume that 
the order and the content of the adjacency arrays are sorted as defined more precisely
in the next section.


%
\section{Representation}
To show our results we use different representations of the 
given $n$-vertex graph $G = (V, E)$ with $V = \{1, \ldots, n\}$
that all need the same amount of memory.
We next present 
different graph representations.


In our {\em sorted standard representation} (Fig.~\ref{fig:standardRepresentation}), we first store
the number of vertices and
 a {\em table of pointers} $T$ with one
pointer per vertex that points to the adjacency array of the vertex. Subsequently, we store the total length of the 
adjacency arrays.
We additionally
assume for the sorted standard representation that the adjacency array of vertex $i$ is stored 
before the adjacency array of vertex $i+1$ for all $i=1,\ldots,n-1$ and that
all vertices inside an adjacency array are also 
stored in ascending order.
If the adjacency array of a vertex is not given in ascending order,
then it can be sorted using an in-place linear-time radix sort~\cite{FraMP07}.
However, in this case, we cannot restore to the
representation of the given graph.

This representation is economical in space and implicitly contains the information to compute the 
degree of each vertex $v \in V$. The degree $\deg(v)$ of a vertex $v$ equals the length of its adjacency array, and since the adjacency array of a vertex $v$ is written directly before the adjacency array of a vertex $v + 1$, the degree of $v$ equals the pointer differences of $T[v]$ and $T[v + 1]$ for all $v \in V \setminus \{ n \}$. For the last vertex $v = n$ the degree equals the difference of the pointer $T[v]$ and the total length of the array $n + m + 2$ with $n = A[0]$ and $m = A[n + 1]$.
If a vertex $v \in V \setminus \{ n \}$ has degree zero, then its adjacency array is empty and therefore $T[v]$ and $T[v + 1]$ point at the same position.

For our DFS described subsequently, we require to encode information like the state of visited and unvisited vertices.
To be able to do this we transform the sorted standard representation first into a so-called {\em adjacency-array begin-pointer representation} or short the {\em begin-pointer representation} and finally into a so-called  {\em swapped begin-pointer representation}. 


We obtain the  {\em begin-pointer representation} (Fig.~\ref{fig:beginPointerRepresentation})\full{ (Lemma~\ref{lem:beginpointerrep})}
by taking the sorted standard representation and replacing each vertex name $v$ in the adjacency arrays by a pointer to the beginning of the adjacency array of vertex $v$.
Since a vertex of degree zero does not have an adjacency array, we cannot create a pointer into it.
In this case we keep the vertex name, but we mark such a vertex by replacing its pointer in the table $T$ by a self reference, i.e., set $T[v] = v$.

\full{\begin{lemma}\label{lem:beginpointerrep}
	There is an in-place transformation from the sorted standard representation 
	to the begin-pointer representation that runs in  
	linear time.
\end{lemma}
\begin{proof}
	The begin-pointer representation can be computed very easily. Iterate over
	all adjacency arrays and replace each entry $A[i] = T[A[i]]$, with $i \in \{n + 2, \ldots, n + m + 2\}$.
	Also set $T[v] = v$ for each vertex $v$ of degree zero.
	\qed
\end{proof}}

In the begin-pointer representation we can jump from one adjacency array into another, but lack the ability to find out the vertex name of the adjacency array in constant time if we jump into it using some edge.
To resolve this issue we use the swapped begin-pointer representation (Fig.~\ref{fig:beginPointerVertexName}) where we swap the first adjacency pointer of a vertex $v$ by $v$ and move the pointer stored there into the table $T$ of position~$v$\full{ (Lemma~\ref{lem:swapunswap})}. 
In this representation we are still able to access the moved pointer by a lookup at $T[v]$, and know immediately to which vertex the adjacency belongs to.

\full{\begin{lemma}\label{lem:swapunswap}
	There is an in-place transformation that swaps and unswaps a representation in linear
	time.
\end{lemma}
\begin{proof}
	Clearly, we can swap a representation
	by iterating once through $T$ and setting $T[v] = A[p]$, with $p = T[v]$ for $\forall v \in V : v \ne T[v]$, and setting $A[p] = v$.
	To unswap a representation, iterate over all adjacency arrays to find all the vertex names $v = A[i] : v \ne T[v]$ for $i \in \{n + 2, \ldots, n + m + 2\}$ and reverse the swap by setting $A[i] = T[v]$ and $T[v] = i$.
	\qed
\end{proof}}

\conf{In our full version of the paper~\cite{KamS18full}, we show that in-place linear-time transformations between all kinds of representation exist.}

\begin{figure}[h!]
	\centering
	\resizebox {\columnwidth} {!} {
		\begin{tikzpicture}[
start chain = going right,
node distance = 0pt],
\begin{scope}[local bounding box=scope1]

\field{0}{5}{0}
\node[left=0 of 0](a){$A$};
\sep{0}
\node[below=0 of 0](n){$n$};
\field[khaki]{1}{7}{1}
\field[lightblue]{2}{9}{2}
\field[green]{3}{12}{3}
\field[yellow]{4}{14}{4}
\field[orange]{5}{17}{5}
\nameFields{1}{5}{$T$}
\sep{5}
\field{6}{12}{6}
\node[below=0 of 6](e){$2m$};
\sep{6}
\field[khaki]{7}{2}{7}
\field[khaki]{8}{5}{8}
\field[lightblue]{9}{1}{9}
\field[lightblue]{10}{3}{10}
\field[lightblue]{11}{4}{11}
\field[green]{12}{2}{12}
\field[green]{13}{4}{13}
\field[yellow]{14}{2}{14}
\field[yellow]{15}{3}{15}
\field[yellow]{16}{5}{16}
\field[orange]{17}{1}{17}
\field[orange]{18}{4}{18}

\end{scope}
\begin{scope}[-{Stealth[length = 2.5pt]}]%
\end{scope}
\end{tikzpicture}
	}
	\caption{Sorted standard representation of a graph with $m$ undirected or $2m$ directed 
		edges.}\label{fig:standardRepresentation}
	\vspace{4mm}
	
%
	\resizebox {\columnwidth} {!} {
	\begin{tikzpicture}[
	start chain = going right,
	node distance = 0pt],
	\begin{scope}[local bounding box=scope1]

	\field{0}{5}{0}
	\node[left=0 of 0](a){$A$};
	\sep{0}
	\node[below=0 of 0](n){$n$};
	\field[khaki]{1}{7}{1}
	\field[lightblue]{2}{9}{2}
	\field[green]{3}{12}{3}
	\field[yellow]{4}{14}{4}
	\field[orange]{5}{17}{5}
	\nameFields{1}{5}{$T$}
	\sep{5}
	\field{6}{12}{6}
	\node[below=0 of 6](e){$2m$};
	\sep{6}
	\field[khaki]{7}{9}{\textbf{7}}
	\field[khaki]{8}{17}{8}
	\field[lightblue]{9}{7}{\textbf{9}}
	\field[lightblue]{10}{12}{10}
	\field[lightblue]{11}{14}{11}
	\field[green]{12}{9}{\textbf{12}}
	\field[green]{13}{14}{13}
	\field[yellow]{14}{9}{\textbf{14}}
	\field[yellow]{15}{12}{15}
	\field[yellow]{16}{17}{16}
	\field[orange]{17}{7}{\textbf{17}}
	\field[orange]{18}{14}{18}

	\end{scope}
	\begin{scope}[-{Stealth[length = 2.5pt]}]%
	\end{scope}
\end{tikzpicture}
}
\caption{Begin-pointer representation of the graph from
	Fig.~\ref{fig:standardRepresentation}.
	Every adjacency array entry $v$ is replaced with the pointer $p = T[v]$
	to the first position of $v$'s adjacency array.
}\label{fig:beginPointerRepresentation}
\vspace{4mm}
%
%
	\centering
       
       	\centering
       	\resizebox {\columnwidth} {!} {
       		\begin{tikzpicture}[
start chain = going right,
node distance = 0pt],
\begin{scope}[local bounding box=scope1]

\field{0}{5}{0}
\node[left=0 of 0](a){$A$};
\sep{0}
\node[below=0 of 0](n){$n$};
\field[khaki]{1}{9}{1}
\field[lightblue]{2}{7}{2}
\field[green]{3}{9}{3}
\field[yellow]{4}{9}{4}
\field[orange]{5}{7}{5}
\nameFields{1}{5}{$T$}
\sep{5}
\field{6}{12}{6}
\node[below=0 of 6](e){$2m$};
\sep{6}
\field[khaki]{7}{\textbf{1}}{7}
\field[khaki]{8}{17}{8}
\field[lightblue]{9}{\textbf{2}}{9}
\field[lightblue]{10}{12}{10}
\field[lightblue]{11}{14}{11}
\field[green]{12}{\textbf{3}}{12}
\field[green]{13}{14}{13}
\field[yellow]{14}{\textbf{4}}{14}
\field[yellow]{15}{12}{15}
\field[yellow]{16}{17}{16}
\field[orange]{17}{\textbf{5}}{17}
\field[orange]{18}{14}{18}

\end{scope}
\begin{scope}[-{Stealth[length = 2.5pt]}]%
\end{scope}
\end{tikzpicture}
       	}
       	\caption{Swapped begin-pointer representation 
       		of the graph in Fig.~\ref{fig:standardRepresentation}.}
       	\label{fig:beginPointerVertexName}
\end{figure}

\full{It remains to describe how to restore the sorted standard representation (Lemma~\ref{lem:restoresrepresentation}).
If the given representation is not swapped, then make it swapped.
Iterate then over all adjacency arrays and replace each pointer that is not a vertex name by the vertex name it points at.
Finally, unswap the representation and correct the entries of the vertices having degree zero.

\begin{lemma}\label{lem:restoresrepresentation}
	There is an in-place transformation from the begin-pointer
	representation to the sorted standard representation that runs in  
	linear time.
\end{lemma}
\begin{proof}
	In the first step replace the pointers in the adjacency entries by the vertex name they point at, i.e., for all $i \in \{n + 2, \ldots, n + m + 2\}$ with $n < A[i]$ set $A[i] = A[A[i]]$.
	Now do the same in the array $T$, i.e., for all $i \in \{1,\ldots, n\}$ set $T[i] = A[T[i]]$.
	At this point all the pointers are replaced by vertex names and it remains to unswap the representation.
	Iterate over all adjacency arrays and, beginning with the first vertex $v  = 1, \ldots, n$ with $T[v] \ne v$, look for a position $p$ with $v = A[p]$ and set $A[p] = T[v]$ and $T[v] = p$.
	Now it remains to restore the vertices of degree zero, which we do by iterating with $i = \{1, \ldots, n\}$ over $T$ and remember the last $i' = i$ with $T[i] > n$.
	Whenever encountering an entry $T[i] = i$ set $T[i] = i'$.
	\qed
\end{proof}}

\section{Depth-First Search}
Usually a DFS is only an algorithmic scheme how a graph can be explored step by step and does nothing useful.
Its usefulness comes in combination with additional computational steps that are defined by a user for a specific application.
These steps can be encapsulated in functions that we call {\em user-implemented functions}.

To introduce the user-implemented functions \texttt{pre-} and \texttt{postprocess} as well as
 \texttt{pre-} and \texttt{postexplore} we start to sketch their usage in a standard DFS.
Initially all vertices of a graph are unvisited, also called {\em white}.
The algorithm starts by visiting a start vertex $u$.
Whenever a DFS visits a vertex $u$ for the first time it colors $u$ {\em gray} to mark it as visited and executes \texttt{preprocess}$(u)$.
For each outgoing edge $(u,v)$ of $u$, it first calls \texttt{preexplore}$(u,v)$ and
second visits vertex $v$ if $v$ is white.
When finally 
$v$ has no outgoing white neighbors, it marks $v$ as done by coloring it
{\em black} and calls \texttt{postprocess}($v$) and backtracks to 
the parent~$u$.
After backtracking from $v$ to $u$ the algorithm calls \texttt{postexplore}($u, v$).

By using suitable implementations for the four user-implemented functions, the user knows exactly how the exploration takes place and can easily output, e.g., the vertices in pre-, post-, or inorder with respect to the constructed DFS tree.
Not every DFS algorithm supports all these functions. Thus, we can also measure the usefulness of a DFS implementation by the number of supported functions.

To obtain a linear-time in-place DFS on directed graphs, we cannot support calls of 
the functions \texttt{preexplore} and \texttt{postexplore}, which are often not necessary, i.e., to compute pre- 
and post-order.


We now start the description of our DFS algorithm where we expect the graph
being given in the swapped begin-pointer representation.  Our goal is to
encode two information in the representation, but with the knowledge that we
have to restore the representation later.  First, we need to encode the
color of each vertex.  Instead of encoding all three colors we use only the
colors white and {\em gray-black} (as gray or black).  Second, we require to
encode the path that we took to reach a vertex such that we are able to
backtrack to a parent vertex and continue the exploration from there.

For simplicity, we \full{first }assume that every vertex of the directed graph has at least two
neighbors, and we so can conclude that
every pointer in the adjacency arrays points at a position
storing a vertex name $v \in V = \{1, \ldots, n\}$.
\full{Afterwards we show how to handle degree zero and one vertices.}
\conf{In our full version of the paper~\cite{KamS18full}, we describe how to handle vertices of degree one and zero.}

\full{\subsection{Handling Vertices of Degree at Least Two}}
Our idea is to store the colors of the vertices implicitly by using the following invariant: 
A vertex $v$ is white exactly if the first pointer $p$ in the adjacency array of $v$, which is stored in $T[v]$, points at a value at most $n$, i.e., $A[p] \le n$. By our conclusion this is initially true for all vertices.
\full{

}%
%
We next want to enable the algorithm to backtrack from a visited vertex to its parent.
Whenever a DFS takes a path from a vertex $u$ to a vertex $v$ it has to return to the vertex $u$ from $v$, i.e., backtrack from $v$ to $u$, if all white neighbors of $v$ are visited.
Our idea is to reverse the path from vertex $u$ to the vertex $v$ whenever we visit a white vertex $v$ by using so-called {\em reverse pointers}.
In other words, the idea is to turn
the pointer to $v$ in $u$'s adjacency array to a pointer to $u$ in $v$'s adjacency array.

Now we describe the construction of a reserve pointer in detail.
See also Fig.~\ref{fig:someLabel}.
Assume that our DFS currently visits a vertex $u$, and we iterate through $u$'s adjacency array.
Iterating over $u$'s adjacency array, e.g., at a position $p$, we find a pointer $q$ pointing into an adjacency array of a white vertex $v = A[q]$.
Inside $v$'s adjacency array the first pointer that we have to inspect is $q' = T[v]$.
Because we know that we left from position $p$ to $q$ to reach $v$, we want to store a pointer to $p$ as a reverse pointer from $v$ to $u$. (Returning to $u$, the algorithm can continue exploring $u$'s adjacency array from $p + 1$.)
We store $p$ inside $T[v]$. 
The pointer $p$ is now the reverse pointer from $v$ to $u$. 
Naively doing so we overwrite the pointer $q'$.
This would cause an information loss. Therefore, we have to find a new location for $q'$.
What we can observe is that when using the reverse pointer, we can restore the original pointer from $u$ to $v$ such that we do not need to keep the pointer $q$ in $A[p]$ (part of $u$'s adjacency array) as long as we have the reverse pointer. Hence, we use $A[p]$ as a temporary location to store $q'$.
Note that $q'$ is still accessible from $v$ by following the reverse pointer stored in $T[v]$.

In the example above we showed how to visit a vertex from a position $p$.
If $p$ is not the first position of $u$'s adjacency array the creation of a reverse pointer that points at $p$ has a nice side-effect: The vertex $v$ becomes gray-black since the value stored in $T[v]$ points at a value larger than $n$.

%
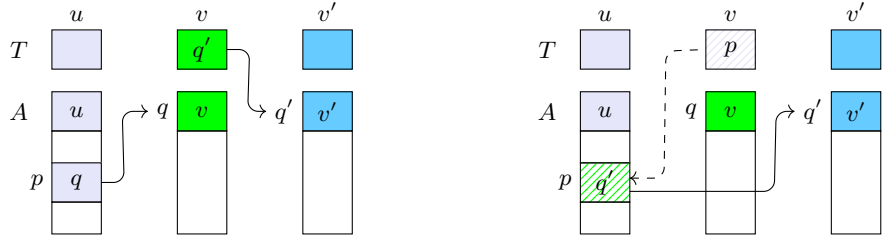
\begin{figure}[t!]%
	\centering
	\begin{tikzpicture}[
	start chain = going right,
	node distance = 0pt],
	\begin{scope}[local bounding box=scope1]
	
		\node [box, fill=purple] (a1) {$u$};
		\node [left=0.2cm of a1] (lla) {$A$};
		\node [box, minimum height=1.3em, below=0 of a1] (a2) {};
		\node [box, below=0 of a2, fill=purple] (a3) {$q$};
		\node [left=0 of a3] (ap) {$p$};
		\node [box, minimum height=1.3em, below=0 of a3] (a4) {};
		
		\node [box, right=1cm of a1,fill=green] (b1) {$v$};
		\node [left=0 of b1] (qv) {$q$};
		\node [box, minimum height=4.2em, below=0 of b1] (b2) {};
		
		\node [box, right=1cm of b1, fill=lightblue] (c1) {$v'$};
		\node [left=0 of c1] (qv2) {$q'$};
		\node [box, minimum height=4.2em, below=0 of c1] (c2) {};
		
		\path [->, draw, rounded corners](a3.east) -- ++(0.3, 0.0) -- (0.6, 0.0) -- (qv.west);

		\node [box, above=0.3cm of a1,fill=purple] (t1) {};
		\node [left=0.2cm of t1] (t0) {$T$};
		\node [box, above=0.3cm of b1,fill=green] (t2) {$q'$};
		\node [box, above=0.3cm of c1, fill=lightblue] (t3) {};
		\node [above=0cm of t1] (l1) {$u$};
		\node [above=0cm of t2] (l2) {$v$};
		\node [above=0cm of t3] (l3) {$v'$};
		
		\path [->, draw, rounded corners](t2.east) -- ++(0.3, 0.0) -- (2.25, 0.0) -- (qv2.west);
		// -----

		\node [box, right=3cm of c1, fill=purple] (d1) {$u$};
		\node [left=0.2cm of d1] (lda) {$A$};
		\node [box, minimum height=1.3em, below=0 of d1] (d2) {};
		\node [box, below=0 of d2, pattern=north east lines,pattern color=green] (d3) {$q'$};
		\node [left=0 of d3] (ap2) {$p$};
		\node [box, minimum height=1.3em, below=0 of d3] (d4) {};
		
		\node [box, right=1cm of d1, fill=green] (e1) {$v$};
		\node [left=0 of e1] (qv3) {$q$};
		\node [box, minimum height=4.2em, below=0 of e1] (e2) {};
		
		\node [box, right=1cm of e1,fill = lightblue] (f1) {$v'$};
		\node [left=0 of f1] (qv4) {$q'$};
		\node [box, minimum height=4.2em, below=0 of f1] (f2) {};

		\path [->, draw, rounded corners](d3.340) -- ++(1.9, 0.0) -- (9.2, 0.0) -- (qv4.west);
	
		\node [box, above=0.3cm of d1, fill=purple] (t11) {};
		\node [left=0.2cm of t11] (t00) {$T$};
		\node [box, above=0.3cm of e1, pattern=north east lines,pattern color=purple] (t22) {$p$};
		\node [box, above=0.3cm of f1, fill=lightblue] (t33) {};
		\node [above=0cm of t11] (l11) {$u$};
		\node [above=0cm of t22] (l22) {$v$};
		\node [above=0cm of t33] (l33) {$v'$};
		
		\path [<-, draw, dashed, rounded corners](d3.10) -- ++(0.45, 0.0) -- (7.75, 0.8) -- (t22.west);
	
	\end{scope}
	\begin{scope}[-{Stealth[length = 2.5pt]}]%
	\end{scope}
	\end{tikzpicture}
	\caption{The figure shows the state before (left) and after (right)
creating a reverse pointer.  The squares at the top are locations in $T$ and
the array bellow of each square of $T$ is the adjacency array of the vertex
written on the top.  The variables $u, v, v' $ are
vertices and $p, q, q' \ge n + 2$ are array positions / pointers.  Normal
arrows are pointers from an adjacency array into another and dashed
arrows are reverse pointers.}
%
	\label{fig:someLabel}
\end{figure}%
%


What if $p$ is the first position in $u$'s adjacency array?
Then we encounter two problems.
To handle the problems, recall that a reverse pointer of a vertex $v$ is always stored in $T[v]$.
In this scenario the reverse pointer $p = T[v]$ points to the first position of an adjacency array that stores a vertex name $u = A[p]$.
The first problem is that $v$ is no longer white because $p$ is the position of a value at most $n$.
The second problem arises when we try to temporary store the pointer $q' = T[v]$ to $A[p]$,
which stores the vertex name $u$ in our swapped representation. Alternatively, storing the pointer $q'$ in $T[u]$ overwrites the reverse pointer of vertex $u$, unless $u$ is the start vertex.

We avoid both problems by never leaving a vertex from the first position of its adjacency array.
If we have to visit a vertex by following the first pointer stored at the first position $p$, i.e, stored in $T[u]$ with $u = A[p]$, then we first swap the pointers in $T[u]$ and $A[p + 1]$ and follow afterwards the pointer stored at the second position $p + 1$.
Since the pointers in our adjacency arrays are stored in ascending order, we can check if we have swapped pointers.
Whenever we return
to a vertex that we left from 
a second position $p$ in its adjacency array 
and the value stored at $p$ is smaller than the value 
in $T[u]$ with $u = A[p - 1] \land 1 \le u \le n$, we swap 
the pointers in
$A[p]$ 
and $T[u]$ back, and follow the pointer at position $p$ to the second vertex.
This ensures that we never leave from the first adjacency position of a vertex and thus never have to store a reverse pointer pointing to a first adjacency position.

We have shown how to create reverse pointers; now it remains to describe how to remove them again. 
After exploring every neighbor of a vertex $v$, our algorithm finds the start of the adjacency array of vertex $v''$, i.e., we find a position $q''$ with $1 \le A[q''] \le n$ (or $q''$ is the end of the whole array $A$).
Note that $v'' = v + 1$, but we do not know $v$ at this point and thus, we cannot search for $v + 1$.
Now, we need to backtrack and thus find the reverse pointer of~$v$.
We find the reverse pointer $p = T[v]$ by iterating backwards until we find a position $q$ with $A[q] \le n$. In fact, then $A[q] = v$.
Now we move the temporary stored pointer $q' = A[p]$ into $T[v]$ again,
and restore the original pointer to $v$ at position $p$ by setting $A[p] = q$. However, this turns $v$ into a white vertex again,
which we solve by incrementing the first 
pointer $q' = T[v]$ of $v$ by one such that the pointer
points to a position storing a value larger than $n$.
Since we assume a degree of at least two for all vertices the incrementation has the effect that the pointer points at a value strictly greater than $n$. The incrementation is easily reversible such that the restoration is trivial.

%
%


Before we present the remaining details of our algorithm, we summarize the
possible modifications in $T$ and the adjacency arrays of the vertices in the 
following three invariants that hold before and after each call of 
\textsc{follow} and \textsc{backtrack}. Before, note that the only other operation
that changes values is \textsc{nextNeighbor}, which only swaps adjacency
pointers, but does not change colors of vertices and the invariants are not affected.

\begin{enumerate}
	\item A vertex $v$ is white exactly if $v$ is not a start vertex and $1 \le A[T[v]] \le n$.
	\item Every gray-black vertex $v$ on a current DFS path, except the start vertex, stores the reverse pointer at $T[v]$ that points into its parent adjacency array at a position $p = T[v]$ with $A[p] \ge n$. Moreover, $p$ is the position where the parent of $v$ originally stored the pointer to $v$.
	\item The first pointer $q = T[v]$ in the adjacency array of a gray-black vertex $v$ that is not on the current DFS path points with its first pointer $q = T[v]$ to the second position $q'$ of another vertex adjacency array, i.e.,  $1 \le A[q' - 1] \le n$.
\end{enumerate}

In detail, our DFS runs as follows. If a start-vertex $1 \le v_{\rm{s}} \le n$ is given, we search for the first position $p$ with $v_{\rm{s}} = A[p]$ of its adjacency array in $O(m)$ time.
Alternatively, we search for a position $p$ with $v_{\rm{s}} = A[p] \land 1 \le v_{\rm{s}} \le n$.
Then, we call \textsc{visit}($p$) that is described now.


\begin{itemize}
	\item \textsc{visit}($p$): (Visit the vertex whose adjacency array starts at position $p$.)
	In the swapped begin pointer representation, $v = A[p]$ is always 
	the vertex name.
	First, call \texttt{preprocess}($v$).
	Finally, start iterating through the neighbors starting from position $p$ by
	executing 
	$\textsc{nextNeighbor}$($p, \textsc{true}$).
	\item \textsc{nextNeighbor}($p, {\textsc{ignoreCheck}}$): (Follows the edge at position $p$
    if the opposite endpoint of the edge is white. Otherwise, it tries the position~$p+1$.)
    
	First of all, we test if $p$ is the 
        first position in the current          
        adjacency array or two position after it 
        by determining if $(\lnot 
        \textsc{ignoreCheck}\ \land (1 \le A[p] \le n))$ or if $1 \le A[p - 2] \le n$, respectively.
    If so, define $p'$ (and $p''$) such that $p'$ is the first ($p''$ is the second) position in the adjacency array and check additionally if the first pointer 
    (which is temporary stored in a parent vertex in $A[r]$ with $r = T[u], u = A[p']$), and the second pointer in $A[p'']$ 
    are swapped, which means that the first is larger than the second pointer.
    Use the information computed above and proceed with Substep~1.
\begin{enumerate}[wide, labelwidth=!, labelindent=0pt]
	\item[\em Substep 1.] If $p$ is the first entry, increment $p$ by one, swap the two
	pointers in $A[r]$ and $A[p'']$
	as well as proceed with Substep 3 to 
         visit the first neighbor
	(if white) from the second position of the adjacency array.

	If $p$ is two positions after the first entry and the two pointers are swapped, (i.e.,
        we just returned from the first neighbor), decrement $p$ by one, swap the
        two pointers as described above and also proceed with Substep 3 to 
	visit the second neighbor (if white) from the second position of the adjacency array.

        Otherwise, we just returned from the second, third, etc.\ neighbor. 
        Then, we go to Substep 2 to test if we reached the end of the current 
        adjacency array and then proceed with Substep 3.
	
	\item[\em Substep 2.] 
        We check if we require to backtrack, i.e, we 
	reached 
	the next 
	adjacency array or are out of 
	index in array $A$. Hence, check if $(1 \le A[p] \le n) \lor (p > n + m + 2)$.
	If we have to backtrack, search for the largest position $q < p$ such that $1 \le A[q] \le n$ and call 
	\textsc{backtrack}($q$) unless $A[q] = v_{\rm{s}}$.
	In that case color $v_{\rm{s}}$ gray-black by incrementing its firs adjacency pointer $T[v_{\rm{s}}]$ by one.
	We now have to explored everything reachable from $v_{\rm{s}}$.
	If wanted, start a new DFS with a next white vertex.
	
	\item[\em Substep 3.] Check if the edge at $p$ points to a white vertex $v = A[q]$ with $q = A[p]$ by running
	the non-recursive procedure \textsc{isWhite}($v$).
	If $p$ does, call $\textsc{follow}(p)$.
	Otherwise, call \textsc{nextNeigbor}$(p + 1, \textsc{false})$.
\end{enumerate}

	\item \textsc{isWhite}($v$): (Return \textsc{true} exactly if the vertex $v$ is white.) 
	We check the first invariant, i.e., return $v \ne v_{\rm{s}} \land 1 \le A[T[v]] \le n $.
	\item \textsc{follow}($p$): (Discover a new child via an edge $e$ stored at 
	position $p$ and color the new discovered vertex implicitly gray-black.)
	First we determine the position $q =A[p]$ and the vertex $v = A[q]$ where $e$ points to. 
	Second, we are going to create a reverse pointer in $T[v]$ to backtrack later. To 
	not lose the pointer previously stored in $T[v]$ we store it in $A[p]$. In detail, remember 
	the first pointer $x = T[v]$ of the neighbor.
	Now, store the pointer inside $A[p] = x$ and create a 
        reverse pointer from the neighbors first adjacency entry 
	into 
	its parent's adjacency array by setting $T[v] = p$.
	Finally, visit the neighbor by executing \textsc{visit}($q$).
	\item \textsc{backtrack}($q$): (From a child $v$ go to its parent 
        where 
        $q$ is the
        beginning of $v$'s adjacency array and
      $p = T[v]$ with $v = A[q]$
        is a reverse pointer to the adjacency array of the parent.)
	Before going to the parent, we have to restore the edges that we
 modified by visiting $v$ such that we fulfill the third invariant.
	In detail, we first restore the child's edge that was temporarily stored in the
 parent's adjacency array, but let it point one edge
	further to guarantee the third invariant. Thus, we set $T[v] = A[p] + 1$ and $A[p] = q$ with $v = A[q]$ and $p = T[v]$.
	Finally, we call \texttt{postprocess}($v$) and subsequently $\textsc{nextNeighbor}(p + 1, \textsc{false})$.
	
\end{itemize}
%

Concerning the running time on $n$-vertex $m$-edge directed graphs, we can observe that all functions of our
in-place DFS run in constant time per call. Moreover, \textsc{visit} and
\textsc{backtrack} are called $O(n)$ times whereas all other
\mbox{functions} are
called $O(m)$ times. Thus, 
our in-place DFS runs in $O(n+m)$ time.
Ignoring the calls for the user-defined functions as well as for
\textsc{isWhite}, which is not recursive,
we only make tail-calls and consequently require no recursion stack.

\full{\subsection{Handling Vertices of Degree Zero}}
\full{
%
We now focus on a vertex $v$ of degree zero.
For an illustration see Fig.~\ref{fig:visitDegZero}.
The only operation that we can do after visiting $v$ is to backtrack.
Assume that we discover $v$ from a vertex $u$ of degree at least two from position $p$.
We call $\texttt{preprocess}(v)$ and $\texttt{postprocess}(v)$.
Now it remains to mark $v$ as gray-black to avoid visiting it over other possible incoming edges.
We 
define 
a vertex of degree zero as white 
if $T[v] = 
v$ holds. Otherwise, $v$ is gray-black.
Whenever we visit $v$,  we create a reverse pointer to $u$ by 
setting $T[v] = p$---similar as we did for vertices 
of degree at least two---and so turn $v$ gray-black.
In contrast to vertices of degree at least two, we do not remove the reverse pointer when 
backtracking from $v$. Instead, we have to run a restoration after the DFS. 
Moreover, even if $v$ was discovered from a swapped pointer in $u$, 
we do not change the reverse pointer stored in $T[v]$ 
when unswapping the pointers, i.e., the reverse pointer never points to a first entry of an adjacency 
array. This helps to identify the reverse pointer during the restoration.

If $v$ is visited from a vertex $u$ of degree one, then $u$ has only one adjacency entry.
This means that we left $u$ from its first adjacency position $p$.
Creating a reverse pointer in this case will mark $v$ as white.
Instead of storing a reverse pointer for $v$ in $T[v]$ 
we (1) extend the first invariant such that 
a vertex $v$ is white if additionally $(T[v] = v)$ holds and (2)
make $v$ gray-black by storing the vertex name $u$ in $T[v]$, 
i.e., set $T[v] = A[p]$ instead of $T[v] = p$. 


\begin{figure}[t!]
	\centering
	\begin{tikzpicture}[
	start chain = going right,
	node distance = 0pt],
	\begin{scope}[local bounding box=scope1]
	
	\node [box, fill=purple] (a1) {$u$};
	\node [box, minimum height=1.3em, below=0 of a1] (a2) {};
	\node [box, below=0 of a2] (a3) {$v$};
	\node [left=0 of a3] (ap) {$p$};
	\node [box, minimum height=1.3em, below=0 of a3] (a4) {};

	\node [box, above=0.2cm of a1,fill=purple] (t1) {};
	\node [left=0.2cm of t1] (t0) {$T$};	
	\node [box, right=0.2cm of t1, fill=green] (t2) {$v$};
	\node [above=0cm of t1] (l1) {$u$};
	\node [above=0cm of t2] (l2) {$v$};
	
	\draw[->] (t2.340-50) arc (180:190+264:2.5mm);
	\path [->, draw, rounded corners](a3.east) -- ++(0.55, 0.0) -- (t2.south);
	
	
	\node [box, fill=purple, right=2.4cm of a1] (a11) {$u$};
	\node [box, minimum height=1.3em, below=0 of a11] (a22) {};
	\node [box, below=0 of a22] (a33) {$v$};
	\node [left=0 of a33] (ap1) {$p$};
	\node [box, minimum height=1.3em, below=0 of a33] (a44) {};

	\node [box, above=0.2cm of a11,fill=purple] (t11) {};
	\node [left=0.2cm of t11] (t00) {$T$};	
	\node [box, right=0.2cm of t11, pattern=north east lines,pattern color=purple] (t22) {$p$};
	\node [above=0cm of t11] (l11) {$u$};
	\node [above=0cm of t22] (l22) {$v$};
	
	\path [->, draw, rounded corners](a33.east) -- ++(0.55, 0.0) -- (t22.south);
	\path [<-, dashed, draw, rounded corners](a33.20) -- ++(0.4, 0.0) -- (t22.240);
	

	\node [box, fill=purple, right=2.5cm of a11] (a111) {$u$};
	\node [left=0 of a111] (ap11) {$p$};
	
	\node [box, above=0.2cm of a111,fill=purple] (t111) {$v$};
	\node [left=0.2cm of t111] (t000) {$T$};	
	\node [box, right=0.2cm of t111, fill=green] (t222) {$v$};
	\node [above=0cm of t111] (l111) {$u$};
	\node [above=0cm of t222] (l222) {$v$};
	
	\path [->, draw, rounded corners](a111.340) -- ++(0.525, 0.0) -- (t222.south);
	\draw[->] (t222.340-50) arc (180:190+264:2.5mm);
	
	
	\node [box, fill=purple, right=2.4cm of a111] (a1111) {$u$};
	\node [left=0 of a1111] (ap111) {$p$};
	
	\node [box, above=0.2cm of a1111,fill=purple] (t1111) {$v$};
	\node [left=0.2cm of t1111] (t0000) {$T$};	
	\node [box, right=0.2cm of t1111, pattern=north east lines,pattern color=purple] (t2222) {$u$};
	\node [above=0cm of t1111] (l1111) {$u$};
	\node [above=0cm of t2222] (l2222) {$v$};
	
	\path [<-, dashed,bend right] (t1111.south) edge (t2222.south);
	\path [->, draw, rounded corners](a1111.340) -- ++(0.525, 0.0) -- (t2222.south);
	
	\end{scope}
	\end{tikzpicture}
	\caption{The two left and two right figures show the states 
		of the representation before and after exploring a 
		vertex $v \in V$ of degree zero from a vertex $u \in V$. On the left side $u$ has degree at least two and on the 
		right $u$ has degree one.}
	\label{fig:visitDegZero}
\end{figure}
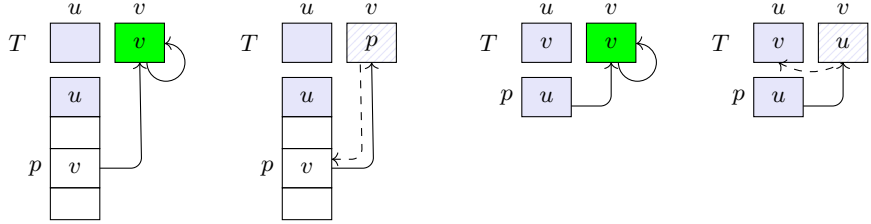}
\full{
\hphantom\\
}

\full{\subsection{Handling Vertices of Degree One}}
\full{We now focus on vertices of degree one.
When we are about to discover such a white vertex $v$ from a vertex $u$ of degree at least two.
Let $p$ be the position of the edge to $v$ in the adjacency array of $u$. 
We can visit $v$ and create a reverse pointer to $u$ by setting $T[v] = p$.
But there is a problem if we want to visit another degree one vertex $v'$
from~$v$:
we have to leave $v$ from 
the first position in $v$'s adjacency array. 

\begin{figure}[t!]
	\centering
	
	\begin{tikzpicture}[
	start chain = going right,
	node distance = 0pt],
	\begin{scope}[local bounding box=scope1]
	
	\node [box, fill=purple] (a1) {$u$};
	\node [box, minimum height=1.3em, below=0 of a1] (a2) {};
	\node [box, below=0 of a2] (a3) {$q$};
	\node [left=0 of a3] (ap) {$\bar{p} = p$};
	\node [box, minimum height=1.3em, below=0 of a3] (a4) {};
	
	\node [box, right=0.7cm of a1,fill=green] (b1) {$v$};
	\node [above=0 of b1] (qv) {$q$};
	
	\node [box, right=0.7cm of b1, fill=lightblue] (c1) {$v'$};
	\node [above=0 of c1] (c2) {$q'$};
	
	\node [box, right=0.7cm of c1, fill=orange] (x1) {$v''$};
	\node [above=0 of x1] (x2) {$q''$};
	\node [box, minimum height=4.2em, below=0 of x1] (x3) {};
	
	\path [->, draw, rounded corners](a3.east) -- ++(0.3, 0.0) -- (0.65, 0.0) -- (b1.west);
	
	
	\node [box, above=0.6cm of a1,fill=purple] (t1) {};
	\node [left=0.2cm of t1] (t0) {$T$};
	\node [box, above=0.6cm of b1, fill=green] (t2) {$q'$};
	\node [box, above=0.6cm of c1, fill=lightblue] (t3) {$q''$};
	\node [box, above=0.6cm of x1, fill=orange] (t4) {$q'''$};
	\node [above=0cm of t1] (l1) {$u$};
	\node [above=0cm of t2] (l2) {$v$};
	\node [above=0cm of t3] (l3) {$v'$};
	\node [above=0cm of t4] (l4) {$v''$};
	
	\path [->, draw, rounded corners](t2.east) -- ++(0.3, 0.0) -- (2, 0.0) -- (c1.west);
	\path [->, draw, rounded corners](t3.east) -- ++(0.3, 0.0) -- (3.35, 0.0) -- (x1.west);
	
	// -----
	
	\node [box, fill=purple, right=1.6cm of x1] (a11) {$u$};
	\node [box, minimum height=1.3em, below=0 of a11] (a22) {};
	\node [box, below=0 of a22, pattern=north east lines,pattern color=orange] (a33) {$q'''$};
	\node [left=0 of a33] (ap1) {$\bar{p} = p$};
	\node [box, minimum height=1.3em, below=0 of a33] (a44) {};
	
	\node [box, right=0.7cm of a11,fill=green] (b11) {$v$};
	\node [above=0 of b11] (qv1) {$q$};
	
	\node [box, right=0.7cm of b11, fill=lightblue] (c11) {$v'$};
	\node [above=0 of c11] (c22) {$q'$};
	
	\node [box, right=0.7cm of c11, fill=orange] (x11) {$v''$};
	\node [above=0 of x11] (x22) {$q''$};
	\node [box, minimum height=4.2em, below=0 of x11] (x33) {};
	
	\node [box, above=0.6cm of a11,fill=purple] (t11) {};
	\node [left=0.2cm of t11] (t00) {$T$};
	\node [box, above=0.6cm of b11, pattern=north east lines,pattern color=purple] (t22) {$p$};
	\node [box, above=0.6cm of c11, pattern=north east lines,pattern color=green] (t33) {$v$};
	\node [box, above=0.6cm of x11, pattern=north east lines,pattern color=lightblue] (t44) {$v'$};
	\node [above=0cm of t11] (l11) {$u$};
	\node [above=0cm of t22] (l22) {$v$};
	\node [above=0cm of t33] (l33) {$v'$};
	\node [above=0cm of t44] (l44) {$v''$};
	
	\path [<-, draw, dashed, rounded corners](a33.east) -- ++(0.3, 0.0) -- (6.95, 1.1) -- (t22.west);
	\draw[<-, dashed] (t33) edge (t44);
	\draw[<-, dashed] (t22) edge (t33);

	\end{scope}
	\begin{scope}[-{Stealth[length = 2.5pt]}]%
	\end{scope}
	\end{tikzpicture}
	\caption{Left: A path $u, v, v', v'' \in V = \{1, \ldots, n\}$ with $u$ and $v''$ as vertices of degree at least two and $v$ and 
		$v'$ of degree one. Right: Situation after visiting every vertex on the path $(u, v, v', v'')$.
		The first adjacency entry of each vertex is the name of the predecessor or a pointer in its adjacency array.
		The first adjacency pointer of $v''$ is stored at $\bar{p}$.}
	\label{fig:oneChain}
\end{figure}
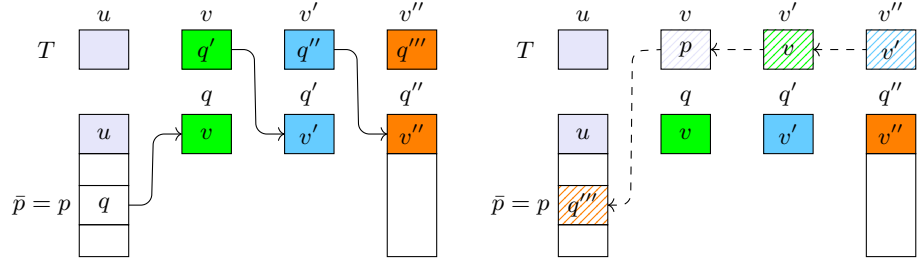

What we can observe is that the only proceeding step after visiting a
vertex $v$ of degree one is to follow $v$'s outgoing edge to the next white
vertex or, if no such edge exists, to backtrack.
Hence, we do not require to visit such a vertex adjacency array again
(because the only existing neighbor is already visited), but need to
backtrack over such a vertex to a previous vertex of degree at least two (or
to the start vertex).
The idea is that vertices visited from vertices of degree one do not store a 
reverse pointer pointing to the position where we 
left from, but store the vertex name of the vertex of degree one where they are visited from.
Having stored the previous vertex enables the algorithm to call \texttt{postprocess} while backtracking over vertices of degree one.

To recognize a vertex of degree one as visited we further extend the first invariant 
to our complete invariant for all vertex degrees: A vertex $v \in V = \{1, \ldots, n\}$ is white exactly if the following equation holds.
\begin{align*}\label{eq:white}
(\underbrace{T[v] = v}_{\deg (v) \rm{\: = \: 0}} \lor \underbrace{T[v] > n}_{\deg (v) \rm{\: = \: 1}}) \land \underbrace{1 \le A[T[v]] \le n}_{\deg (v) \rm{\: \ge \: 2}}
\end{align*}

When backtracking we are not able to restore the pointers, but we restore the pointers after the DFS during an extra restoration
described in the next subsection.

\full{In detail, we handle vertices of degree one as follows:}
Now we consider a vertex $u$ of degree at least two and a position $p$ in $u$'s adjacency array that stores a 
pointer $q = A[p]$ to the adjacency array of a vertex $v=A[q]$ of degree one.
See also Fig.~\ref{fig:oneChain}.
We use a {\em local temporary variable}~$q^*$ to remember the pointer $q' = T[v]$ to
a next white vertex $v' = A[q']$ 
and---as usual---create a reverse pointer 
by setting $T[v] = p$ that points back to the position $p$. 
Moreover, we remember 
in a {\em global
temporary variable} $\bar{p} = p$ 
until we reach a vertex of degree at least two (where we have to replace some pointer $q'''$ by a reverse pointer. Since we do not want to lose $q'''$, we store it
at position $p$---in some sense, we use our usual rule
after contracting induced paths).
Now $v'$ can be of three types: A vertex of degree zero, of degree one, or of degree at 
least two.

A white vertex $v'$ is of degree zero if the condition $A[v'] = v'$ holds.
If not, take $q^*$ as the first position in $v'$s adjacency array.
Then, $v'$ is of degree one exactly if it is not of degree zero and 
$1 \le A[q^*+1] \le n$ holds, i.e., at position $q^* + 1$, a new adjacency array starts.
Otherwise, the vertex is of degree at least two. 

We handle vertices of degree zero as described above.
If $v'$ has degree one, we store the next pointer $q'' = T[v']$ in the local temporary variable $q^*$ and create the reverse pointer $T[v'] =~v$.
In Fig.~\ref{fig:oneChain} the vertex $v$ turns gray-black because the third predicate of our invariant becomes false and $v', v''$ turn gray-black because the second predicate becomes false.
Note that we can not store the pointer $q^*$ inside 
$A[q]$ since it is the first adjacency entry of $v$.

If we reach a vertex $v'' = A[q'']$ of degree at least two, we first
read the pointer $q''' = T[v'']$, remember it in $q^*$ and set a reverse pointer $T[v''] = v'$. 
Now we have to store $q'''$, but not in the previous vertex since it is of degree one. 
Instead, we store it at the remembered position $\bar{p}$ of the previous vertex of degree at least two, i.e., we set $A[\bar{p}] = q^*$ (in the example $q^* = q'''$).

Now, whenever we have to access $q'''$ we have to 
backtrack to the position $\bar{p}$ that stores the pointer. Since we have to access this pointer only two times (whenever we need to compare the first two pointers of a vertex), the running time is still linear.
After visiting a vertex of degree at least two, we can forget pointer $\bar{p}$ again.

It remains to remark that,
if a vertex of degree one is a start vertex, we use a
global variable so that we do not need to store a pointer of another vertex $v$ in its adjacency array to 
create a reverse pointer from a vertex $v$ to the start vertex. 
\\
}

\full{\subsection{Restoration}}
\full{After running the DFS, we need to restore the representation.}
\full{The restoration of vertices of degree at least two is simple. 
Let $v$ be a vertex that points with $T[v]$ into 
the adjacency array of 
a vertex of degree at least two. By the third invariant, $v$
points with its first adjacency position at the second adjacency entry of another vertex, i.e., to restore the 
swapped begin-pointer representation of such a vertex set $T[v] = T[v] - 1$.

It remains to restore 
entries in adjacency arrays that either belong to 
degree-zero vertices or that are part of a chain of degree-one vertices. 
For the restoration of vertices of degree zero, we have to undo the 
changes shown in Fig.~\ref{fig:visitDegZero}.
Every vertex $v$ of degree zero has a reverse pointer into the adjacency position of a vertex $u$ 
from where $v$ was discovered and $u$ still points at $v$, i.e., $u$ and $v$ create a loop or $v$ 
points at a position $p + 1$ where $p$ is the first adjacency position of $u$ (happens if 
$v$ was discovered from the first adjacency position of $u$ that was swapped with the second).

To restore the state of $v$ iterate over the adjacency arrays of all vertices and whenever encountering a position 
$p > n$ with $v = A[p]$ with $v \le n$, we may have found a pointer to an adjacency 
array of degree zero. We found a loop exactly if $T[v] = p \land A[p] \ge n$ ($u$ has degree at least two) or if $p$ is the start of the adjacency array and $v$ points 
at the second position, i.e, $u = 
A[p] \land T[v] = p + 1$, or if $T[v] = u \land  1 \le u \le n$ ($u$ has degree 1). For 
all cases we restore the state by setting $T[v] = v$.

To restore the state of vertices that are involved in a chain of degree-one vertices (recall 
Fig.~\ref{fig:oneChain}), we have to reverse the reverse pointers since we have not done it during 
the backtracking steps of the DFS to keep the vertices gray-black.
To run the restoration we iterate over all adjacency arrays to find a
pointer with a value $v'$ with $1\le v'\le n$ and $v'$ is a vertex of degree
1. Let $v''$ be the vertex whose adjacency array contains the pointer.
%
Then 
follow the reverse pointers to further vertices $v$ of degree one until a
vertex $u$ of degree at least two is reached. 
In each step we reverse the reverse pointer. 
Since we cannot find the right position of a vertex adjacency name, we do not restore the swapped 
begin pointer representation completely. Instead, store only vertex names (instead of pointers to those vertices) such that we 
can harmonize all by computing a sorted standard representation in a next step.
Moreover, move the pointer $q'''$ from $u$ back to $v''$ as shown in
Fig.~\ref{fig:oneChain}.

%
%

After these steps we have restored the direction of the pointers, but still use a vertex name instead of a 
pointer. 
Finally, run a transformation from a begin pointer representation to a sorted standard representation, but 
ignore the entries in the adjacency arrays that are already at most $n$ since these 
are already restored.\\
}

\full{%
	The extensions due to the vertices of degree zero or one do not change 
        the linear asymptotic running: Each such vertex 
        can be handled in $O(1)$ time if we ignore the steps to follow a
        chain of consecutive vertices of degree one
        from a vertex $u$ of degree at
least two to another vertex $v''$ of degree at least two---recall
        Fig.~\ref{fig:oneChain}.
        The chains are used whenever we access 
        $v''$'s first pointer, which is 
        temporary stored in $u$'s adjacency array. This happens only 3
        times (when checking the order of the first and the second pointer originally belonging to $v''$'s adjacency array).
%
To bound the total time used on that chains, 
we can observe that 
the vertices in the chains are disjoint
and therefore the time is $O(n)$.  In a last step we reconstruct the
representation where we iterate a constant number of times 
over the whole array $A$ consisting of $O(n
+ m)$ words.  Altogether, the runtime sums up to $O(n + m)$.\\
}

\begin{theorem}
	 There is an in-place DFS for (un)directed graphs on the weak restore word RAM that runs in $O(n + m)$ time
 on $n$-vertex $m$-edge graphs on our sorted standard representation consisting of $n+m+2$ words ($n + 2m + 2$ words) and
 supports
 calls of the user defined functions \texttt{pre-} and
 \texttt{postprocess}.
\end{theorem}

If $O(n(n + m))$ time is allowed, we can support \texttt{pre-} and
\texttt{postexplore}:
Whenever backtracking from a vertex $v$ to a vertex $u$ we know $v$'s name and return to
a position $p$ in $u$'s adjacency entry. 
Thus, $O(n)$ time allows us to  
lookup the vertex name $u = A[q]$ by searching for the largest $q < p$
with $1 \le A[q] \le n$. 
\section{Breadth-First Search}

As usual for a BFS,
our algorithm runs in rounds and, in round $z-1$ with $z \in \Nat$, all vertices of distance $z$
from a start vertex are added into a new list. Then our algorithm can always iterate
through a list of vertices and for each such vertex~$u$, we iterate through
$u$'s adjacency array. 
%
For a simpler description, 
assume that all vertices are initially white\full{\ and whenever a vertex is added 
into the BFS tree, then it turns light-gray}.\full{\ If we are in the round where
the vertex is processed, the vertex is dark-gray.} 
After adding $u$'s white neighbors
into a list for the next BFS round, the vertex turns black. 

%
To implement our BFS we make use of the following observation.
In the sorted standard representation all words in the table $T$ are stored in ascending 
order.
Our idea is to partition $T$ in regions such that the most significant bits
of the words are equal per region.
We use this to create a \textit{shifted representation} of $T$ by ignoring the 
most significant bits and shifting the words in $T$ together 
(Lemma~\ref{lem:encode}) such that we have a linear number
 of bits free to store a 
$c$-color choice dictionary~\cite{Hag18,HagK16,KamS18c}
as demonstrated in Fig.~\ref{fig:shifted}.
\full{We encapsulate the read access to the words stored packed in $T$ through 
a new data structure $\mathcal{T}$. 
The details of the access are described in the proof of 
Lemma~\ref{lem:encode}.}
%
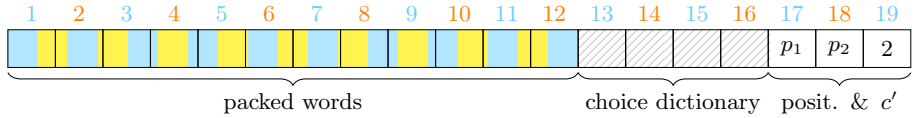
\begin{figure}[h]
	\centering
	\resizebox {\columnwidth} {!} {
		\begin{tikzpicture}[
	start chain = going right,
	node distance = 0pt,
	inner/.style={minimum width=1.265em, minimum height=1.6em,outer 
		sep=0pt}],
	\begin{scope}[local bounding box=scope1]

	\node[fill=lightblue!50, inner] at (-1.2mm, 0) (a) {};
	\node[fill=yellow!80, inner, right = of a] (b) {};
	\node[fill=lightblue!50, inner, right = of b] (c) {};
	\node[fill=yellow!80, inner, right = of c] (d) {};
	\node[fill=lightblue!50, inner, right = of d] (e) {};
	\node[fill=yellow!80, inner, right = of e] (f) {};
	\node[fill=lightblue!50, inner, right = of f] (g) {};
	\node[fill=yellow!80, inner, right = of g] (h) {};
	\node[fill=lightblue!50, inner, right = of h] (a1) {};
	\node[fill=yellow!80, inner, right = of a1] (b2) {};
	\node[fill=lightblue!50, inner, right = of b2] (c3) {};
	\node[fill=yellow!80, inner, right = of c3] (d4) {};
	\node[fill=lightblue!50, inner, right = of d4] (e5) {};
	\node[fill=yellow!80, inner, right = of e5] (f6) {};
	\node[fill=lightblue!50, inner, right = of f6] (g7) {};
	\node[fill=yellow!80, inner, right = of g7] (h8) {};
	\node[fill=lightblue!50, inner, right = of h8] (a2) {};
	\node[fill=yellow!80, inner, right = of a2] (b3) {};
	\node[fill=lightblue!50, inner, right = of b3] (c4) {};
	
	\node[inner, right = of c4, minimum width=8em,very 
	thin,pattern=north east lines,pattern color=gray!40] (cd) {};
	
	\fieldt{1}{}{\color{lightblue}1}
	\fieldt{2}{}{\color{orange}2}
	\fieldt{3}{}{\color{lightblue}3}
	\fieldt{4}{}{\color{orange}4}
	\fieldt{5}{}{\color{lightblue}5}
	\fieldt{6}{}{\color{orange}6}
	\fieldt{7}{}{\color{lightblue}7}
	\fieldt{8}{}{\color{orange}8}
	\fieldt{9}{}{\color{lightblue}9}
	\fieldt{10}{}{\color{orange}10}
	\fieldt{11}{}{\color{lightblue}11}
	\fieldt{12}{}{\color{orange}12}
	\fieldt{13}{}{\color{lightblue}13}
	\fieldt{14}{}{\color{orange}14}
	\fieldt{15}{}{\color{lightblue}15}
	\fieldt{16}{}{\color{orange}16}
	\fieldt{17}{$p_1$}{\color{lightblue}17}
	\fieldt{18}{$p_2$}{\color{orange}18}
	\fieldt{19}{2}{\color{lightblue}19}
	
	\nameFields{1}{12}{packed words\y}
	\nameFields{13}{16}{choice dictionary\y}
	\nameFields{17}{19}{posit.\, \&\, $c'$\y}
	\end{scope}
\end{tikzpicture}	
	}
	\caption{Shifted representation with $c$-color choice dictionary.}
	\label{fig:shifted}
\end{figure}
\begin{lemma}\label{lem:encode}
	Let $c > 0$ be a constant and $n\ge 2^{c+1}w$ be an integer. Having an array of 
        $n$
        ordered words 
	we can pack it in linear time with an in-place algorithm such that we have 
	$cn$ unused bits free  
        and that we still 
	can access all elements of the array in constant time.
        \conf{Afterwards, we can similarly unpack the words.}
        \full{With a similar linear-time in-place algorithm, we can unpack the
        words.}
\end{lemma}
\begin{proof}
	The idea is to partition the array into parts such that each pair of words 
	in a part has the same $c'=c+1$ significant bits.
	Since the sequence is ordered, we iterate over all words and look for the 
	positions where one of the most $c'$ significant bits change.
	\full{During the construction we remember all these $2^{c'}-1$ positions in the working
	memory.} 
	
	Now, the most significant $c'$ bits of each word are equal per region.
	We treat them as unused space. 
%
%
        If we store the remaining $(w - c')$ bits of all words consecutively,
%
        they occupy $n(w - c')$ bits in total such that it leaves
	$c'n$ bits free to use. 
        We use the last 
        $2^{c'}$
        words to store $c'$
        and all the positions. Thus, 
        $c'n - 2^{c'}w \ge c'n - n = cn$
        bits remain free.
	
	For implementing a function \textsc{read}($i \in \{1, \ldots, n\}$) that reads the $i$th
        original word, we have to identify its current 
	position that can be distributed between two words, to cut its bits out of the
	two words and to use the remembered position to reconstruct its most 
	significant bits. For the following description assume that the bits
        of a word are numbered from $0$ (least significant) to $w-1$ (most
        significant). \full{

	}In detail, the $i$th word in $T$ originally stored at bit position $w(i-1)$ was shifted exactly $c(i-1)$ bits
        and now starts after $x=(w-c)(i-1)$ bits, i.e., it starts with bit
        $y=(x\mod w)$ in the word $((x \,{\mathrm{div}}\, w)-1)$ and consists of the next $w-c$
        bits. Using suitable shift operations we can get the $i$th word in constant
        time. To reconstruct its most 
        significant bits, scan over the last $c'$ words to determine the part to
        which $i$ belongs.
%
%
%
%
%
\full{

To restore to the sorted standard representation of the array, 
we store $c'$ and the positions in the working memory. Afterwards, 
we iterate over the words backwards 
and set $T[i] = \textsc{read}(i)$ for all $i\in \{1, \ldots, n\}$.}%
\qed
\end{proof}
%

Before we now obtain our linear-time BFS, we want to remark that the shifted
representation cannot be used to run a standard DFS in-place since a stack for the
DFS can require $\Theta(n\log n)$ bits on $n$-vertex graphs and that many bits are not free in the shifted
representation.

We first prepare the shifted representation of our 
graph (Lemma~\ref{lem:encode}).
Then we can use the free bits to implement a $c$-color choice dictionary
in which we  
store the colors of the vertices, and to iterate over 
colored vertices
in constant time per vertex.
The $c$-color choice dictionary provides the 
following 
functions.
\begin{itemize}
	\item \textsc{setColor}($v, q$): Colors an entry $v$ with the color~$q \in \{0, \ldots, c - 1\}$.
	\item \textsc{color}($v$): Returns the color of the entry~$v$.
	\item \textsc{choice}($q$): Returns an (arbitrary) entry that has the color $q \in \{0, \ldots, c - 1\}$.
\end{itemize}

To start our BFS at vertex $v$, we first initialize a $c$-color choice dictionary 
$D$ for four colors \{{\sc white, light-gray, dark-gray, black}\} with all vertices being initially white. 
Remember in a global variable a round counter $z = 0$ to output the 
round number for each vertex. Then, color the root vertex $v$ light-gray by 
calling $D$.\textsc{setColor}($v, \textsc{light-gray}$).
Finally, we start to process the whole DFS-tree as follows.

Whenever the current round counter $z$ is even, the idea is to iterate over the light-gray 
vertices and color their white neighbors dark-gray and if $z$ is odd 
we do vice versa. 
We next explain the details for the case where $z$ is even.
For an odd $z$, simply switch the words light-gray and dark-gray below.

As long as there is a light-gray vertex $v = D.\textsc{choice}(\textsc{light-gray})$, we output $(v, z)$, color $v$'s white neighbors dark-gray, and color $v$ black.
To color the neighbors we iterate over $v$'s 
adjacency array starting at position $p = \mathcal{T}[v]$ and ending at 
$q = \mathcal{T}[v + 1] - 1$ where we define $\mathcal{T}[n + 1]=n + m + 2$
as the end of our 
graph representation.
For every neighbor $u = A[j]$ with $p \le j \le q$ we check if 
$D.\textsc{color}(u) = \textsc{white}$ and if so, we color $u$ dark-gray by 
calling $D.\textsc{setColor}(u, \textsc{dark-gray})$, otherwise we ignore it.
After the iteration over $v$'s adjacency array we call $D.\textsc{setColor}(v, \textsc{black})$. Since $v$ is now black, the 
next call of $D.\textsc{choice}(\textsc{light-gray})$ returns the next light-gray
\mbox{vertex
if one exists.}

If we could color a vertex dark-gray during the current iteration over the light-gray vertices, 
then 
there are vertices left to process: We increase $z$ by one
and start a new round by iterating now over the 
dark-gray colored vertices as described.
Otherwise, the BFS finishes.

By Lemma~\ref{lem:encode}, we can restore to the sorted standard representation.

\begin{theorem}
There is an in-place BFS for (un)directed graphs on the weak restore word RAM that runs in 
$O(n + m)$ time on $n$-vertex $m$-edge graphs on our sorted standard representation 
consisting of $n+m+2$ words ($n + 2m + 2$ words).
\end{theorem}

\section{Conclusion}
We showed linear-time in-place algorithms for DFS and BFS on the weak restore word RAM 
that have the same asymptotic running time as the standard algorithms. 
To evaluate the usability in practice we implemented the folklore and the
linear-time 
in-place DFS. The implementations are published 
on GitHub~\cite{KamS18g}.


%
%
{Even if we consider our graph representation to be economical in its space requirement, Farzan and 
Munro~\cite{FarM13} showed 
a 
succinct graph representation with constant access-time that requires only $(1 + \epsilon) 
\log {{n^2}\choose{m}}$ bits for any constant $\epsilon > 0$. An interesting open question is if it is 
possible to  
implement 
a (linear-time)
in-place algorithm for DFS or BFS by using the succinct graph representation of Farzan and Munro or 
one that 
requires a little more space.}

\subsection*{Acknowledgments}
Andrej Sajenko was funded by the Deutsche Forschungsgemeinschaft (DFG, German Research Foundation) -- 379157101.


\phantomsection
\addcontentsline{toc}{chapter}{Bibliography}
\bibliography{main}

%
%
%
%
%
%
%
%

\end{document}